\documentclass{llncs}

\usepackage{multicol}

\usepackage{amsmath}
\usepackage{amsfonts}
\usepackage{amssymb}

\usepackage{graphicx}
\usepackage{epic}
\usepackage{eepic}
\usepackage{epsfig,float}
\usepackage{verbatim}
\usepackage{pdfsync}

\newcommand{\ol}{\overline}
\newcommand{\eps}{\varepsilon}
\newcommand{\emp}{\emptyset}

\newcommand{\Sig}{\Sigma}

\newcommand{\noin}{\noindent}

\newcommand{\bi}{\begin{itemize}}
\newcommand{\ei}{\end{itemize}}
\newcommand{\be}{\begin{enumerate}}
\newcommand{\ee}{\end{enumerate}}
\newcommand{\bd}{\begin{description}}
\newcommand{\ed}{\end{description}}
\newcommand{\bq}{\begin{quote}}
\newcommand{\eq}{\end{quote}}

\newcommand{\cA}{{\mathcal A}}

\newcommand{\cD}{{\mathcal D}}

\newcommand{\cP}{{\mathcal P}}

\newcommand{\cS}{{\mathcal S}}

\newcommand{\one}{{\mathbf 1}}

\newcommand{\lraL}{{\hspace{.1cm}{\leftrightarrow_L} \hspace{.1cm}}}

\newcommand{\qedb}{\hfill$\blacksquare$}

\title{Symmetric Groups and Quotient Complexity of Boolean 
Operations\thanks{This work was supported by the Natural Sciences and Engineering Research Council of Canada under grants No.~611456 and~OGP0000871, 
by  the  European Regional Development Fund through the programme COMPETE,    and by the Portuguese Government through the FCT under projects    PEst-C/MAT/UI0144/2011 and CANTE-PTDC/EIA-CCO/101904/2008.
}
}

\author{Jason Bell\inst{1} \and
Janusz~Brzozowski\inst{2} \and 
Nelma Moreira\inst{3} \and Rog\'erio Reis\inst{3}}

\authorrunning{Bell, Brzozowski, Moreira, Reis}   

\institute{Department of Pure Mathematics, 
University of Waterloo, \\
Waterloo, ON, Canada N2L 3G1\\
\{{\tt jpbell@uwaterloo.ca}\}  \and 
David R. Cheriton School of Computer Science, University of Waterloo, \\
Waterloo, ON, Canada N2L 3G1\\
\{{\tt brzozo@uwaterloo.ca}\}
\and
CMUP \& DCC, Faculdade de Ci{\^e}ncias da Universidade do Porto,\\
Rua do Campo Alegre, 4169--007 Porto Portugal\\
\{{\tt \{nam,rvr\}@dcc.fc.up.pt}\}
}

\begin{document}

\maketitle
\begin{abstract}
The quotient complexity of a regular language $L$ is the number of left quotients of $L$, which  is the same as the state complexity of $L$.
Suppose that $L$ and $L'$ are binary regular languages with quotient complexities $m$ and $n$,  and that the transition semigroups of the minimal deterministic automata accepting $L$ and $L'$ are  the symmetric groups $S_m$ and $S_n$ of degrees $m$ and $n$, respectively.
Denote by $\circ $ any binary boolean operation that is not a constant and not a function of one argument only.
For $m,n\ge 2$ with $(m,n)\not \in \{(2,2),(3,4),(4,3),(4,4)\}$ we prove that the quotient complexity of $L\circ L'$ is $mn$ if and only either (a) $m\not= n$ or (b) $m=n$ and the bases (ordered pairs of generators) of $S_m$ and $S_n$ are not conjugate.  For $(m,n)\in \{(2,2),(3,4),(4,3),(4,4)\}$ we give examples to show that this need not hold.
In proving these results we generalize the notion of uniform minimality to direct products of automata. We also establish a non-trivial connection between complexity of boolean operations and group theory.
\medskip

\noin
\end{abstract}

\section{Motivation}

The \emph{left quotient,} or simply \emph{quotient,} of a regular language $L$ over an alphabet $\Sig$ by a word $w\in\Sig^*$
is the regular language $w^{-1}L=\{x\in\Sig^*\colon wx\in L\}$.
It is well known that a language is regular if and only if it has a finite number of quotients. 
Consequently, the number of quotients of a regular language, its \emph{quotient complexity,} is a natural measure of complexity of the language.
Quotient complexity is also known as \emph{state complexity,} which is the number of states in the complete minimal \emph{deterministic finite automaton} (\emph{DFA}) recognizing the language.
We prefer quotient complexity because it is a language-theoretic concept, whereas state complexity involves a completely different object, the DFA.
State complexity was first studied by Maslov~\cite{Mas70} in 1970, but it only attracted much interest after 1994 as a result of the paper by Yu, Zhuang and K.\ Salomaa~\cite{YZS94}.
For more details about state complexity see the survey by Yu~\cite{Yu01}. 
The quotient point of view was introduced in 2010 by Brzozowski~\cite{Brz10}.
In this paper we usually refer to quotient/state complexity simply as \emph{complexity}.

The problem of determining the \emph{(quotient) complexity of an operation}~\cite{Brz10,Mas70,Yu01,YZS94} on regular languages has received much attention. It is defined as the maximal complexity of the language resulting  from the operation, taken as a function of the complexities of the operands.
Languages that meet the upper bound on the complexity of an operation are \emph{witnesses} for this operation.
Although witnesses for common operations on regular languages are well known, there are occasions when one has to look for new witnesses:
\be
\item
One may be interested in  a \emph{class} of languages that have the same complexity with respect to a given operation. 
For example, let $\Sig=\{a,b\}$ and let $|w|_a$ be the number of times the letter $a$ appears in the word $w\in\Sig^*$.
Then the intersection of the languages $L=\{w\in\Sig^*\colon|w|_a\equiv m-1 \mbox{ mod } m\}$ and $L'=\{w\in\Sig^*\colon|w|_b\equiv n-1 \mbox{ mod } n\}$ has complexity $mn$. 
The languages $K=(b^*a)^{m-1}\Sig^*$ and $K'=(a^*b)^{n-1}\Sig^*$ also meet this bound~\cite{BJL13}; hence  $(L,L')$ and $(K,K')$  are in the same complexity class with respect to intersection.
\item
Whenever one studies complexity within a  \emph{proper subclass} of regular languages, one usually needs to find new witnesses.
For example, in the class of regular right ideals---languages $L\subseteq\Sig^* $ satisfying $L=L\Sig^*$---languages  $K$ and $K'$ are appropriate, but $L$ and $L'$ are not.
\item
When one studies \emph{combined operations} --- operations that involve more than one basic operation,  for example, the intersection of reversed languages --- once again need new witnesses~\cite{LMSY08}. 
\ee
Before stating our result, we provide some additional background information.
The \emph{Myhill congruence} $\lraL$ of $L$ is defined as follows~\cite{Myh57}: For all $x,y \in \Sig^*$, 
\begin{equation*}
x~\lraL~y \mbox{ if and only if } uxv\in L  \Leftrightarrow uyv\in L\mbox { for all } u,v\in\Sig^*.
\end{equation*}
The set $\Sig^+/ \lraL$ of equivalence classes of the relation $\lraL$ is a semigroup with concatenation as the operation; it is called the \emph{syntactic semigroup} of $L$, which we denote by $S_L$. 
It is well known that the syntactic semigroup is isomorphic to the semigroup $S_\cD$ of transformations performed by  non-empty words on the set of states in the minimal DFA $\cD$ recognizing $L$; this semigroup is known as the \emph{transition semigroup} of $\cD$.
If $\cD$ has $n$ states, the cardinality of the transition semigroup is bounded from above by $n^n$, and this bound is reachable.

The \emph{atoms}~\cite{BrTa11,BrTa12} of a regular language are non-empty intersections of left quotients of the language, some or all of which may be complemented. A regular language has at most $2^n$ atoms, and their quotient complexities are known~\cite{BrTa12}.

The \emph{reverse} of a word is defined inductively: the reverse of the empty word $\eps$ is $\eps^R=\eps$, and the reverse of $wa$ with $w\in\Sig^*$ and $a\in \Sig$ is $(wa)^R=aw^R$.
The reverse of a language $L$ is $L^R=\{w^R\colon w\in L\}$.
The maximal complexity of $L^R$ for $L$ with  complexity $n$ is $2^n$, and this bound is reachable~\cite{Mir66}.
\medskip

Whenever new witnesses are used, it is necessary to prove that these witnesses  meet the required bound. It would be very useful to have results stating that \emph{if the languages in question have some  property $P$, then they meet the upper bound for a given operation.} 
Some results of this type are now now briefly discussed.

Let \textbf{MSC} denote the class of languages with \emph{maximal syntactic complexity} (languages with largest syntactic semigroups), let \textbf{STT} denote the class of languages whose minimal DFAs have \emph{set-transitive transition semigroups} (for any two sets of states of the same cardinality there is a transformation that maps one set to the other), let \textbf{MAL} denote the class of \emph{maximally atomic languages} (languages that have all $2^n$ atoms, all of which have maximal possible quotient complexity), let \textbf{MNA} denote the class of languages with the \emph{maximal number {\rm ($2^n$) of atoms}}, and let \textbf{MCR} denote the class of languages with a \emph{maximally complex reverse} (reverse of complexity $2^n$). The following relations hold~\cite{BrDa13a}:
\[ \text{\textbf{MSC} $\subset$ \textbf{STT} = \textbf{MAL} $\subset$ \textbf{MNA} = \textbf{MCR}.} \]
The fact that $\text{\textbf{MSC} $\subset$  \textbf{MCR}} $
is a result of  A. Salomaa, Wood, and Yu~\cite{SWY04},
and the observation that $ \text{ \textbf{MNA} = \textbf{MCR}}$ was made by Brzozowski and Tamm~\cite{BrTa11}.

 Our main theorem is a similar result for binary boolean operations on regular languages. We say that such a boolean operation is \emph{proper} if $\circ$ is not a constant  and not a function of one variable only. 

Let $S_n$ denote the symmetric group of degree $n$. 
A \emph{basis}~\cite{Pic39} of $S_n$
is an ordered pair $(s,t)$ of distinct transformations of $Q_n=\{0,\dots,n-1\}$ that generate $S_n$.
Two bases $(s,t)$ and $(s',t')$ of $S_n$ are \emph{conjugate} if there exists a transformation $r\in S_n$ such that $rsr^{-1}=s'$, and  $rtr^{-1}=t'$.

We are interested in DFAs whose transition semigroups are symmetric groups.
Assume that a DFA $\cD$ (respectively, $\cD'$) has state set $Q_m$ ($Q_n$), and transition semigroup $S_m$ ($S_n$).
Let $L$ ($L'$) be the language accepted by $\cD$ ($\cD'$).
Our main result is the following:
\begin{theorem}
\label{thm:main}
Let $\cD$ and $\cD'$ be binary DFAs with $m$ and $n$ states respectively, where 
$m,n\ge 2$ and $(m,n)\not \in \{(2,2),(3,4),(4,3),(4,4)\}$.
If  the transition semigroups of $\cD$ and $\cD'$ are $S_m$ and $S_n$ respectively, and $\circ$ is a proper binary boolean operation,
then the quotient complexity of $L\circ L'$ is $mn$, unless
$m=n$ and the bases of the transition semigroups of $\cD$ and $\cD'$ are conjugate, in which case the quotient complexity of $L\circ L'$ is $m=n$.
\end{theorem}

This theorem is a generalization of some results of Brzozowski and Liu~\cite{Brz12,BrLiu12} which will be stated in the next section.
\medskip

The proof that the quotient complexity of a binary boolean operation on two languages is maximal involves two steps. First, one proves that the direct product of the minimal DFAs of the languages is connected, meaning that all of its states are reachable from the initial state.
Second, one verifies that every two states in the direct product are distinguishable by some word, that is, that they are not equivalent.

The remainder of the paper is structured as follows: Section~\ref{sec:term} defines our terminology and notation. Section~\ref{sec:connect} deals with the conditions under which the direct product of two automata is connected. Section~\ref{sec:uniform} studies uniformly minimal semiautomata, that is, semiautomata which become minimal DFAs if one adds an arbitrary set of final states, other than the empty set and the set of all states. 
Section~\ref{sec:main} contains our main result relating symmetric groups to the complexity of boolean operations, for all except a few cases which are dealt with in Section~\ref{sec:small}.
Section~\ref{sec:conc} concludes the paper.

\section{Preliminaries}
\label{sec:term}

\subsection{Groups}
Many results in this paper rely heavily on the theory of finite groups. Here we provide only some basic definitions, and refer the reader to texts on group theory, for 
example~\cite{Rot65,Suz82}, for additional information.

A \emph{semigroup} $(S,\cdot)$ is a set $S$ with an associative binary operation $\cdot$, which we call \emph{multiplication} and often omit. A \emph{monoid} $(S,\cdot,\one)$ is a semigroup with an \emph{identity}
$\one$, which is an element of $S$ satisfying $\one \cdot s = s \cdot \one=s$ for all $s\in S$.
A \emph{group} is a monoid $(G,\cdot,\one)$, such that every element $g\in G$ has an \emph{inverse} $g^{-1}\in G$ that satisfies $g\cdot g^{-1}=g^{-1}\cdot g=\one$.
The \emph{order} of a group $G$ is the number of elements in $G$.
If $g,h\in G$, then $hgh^{-1}$ is a \emph{conjugate} of $g$ (by $h$).

Let $(G,\cdot,\one_G)$ and $(H,\ast,\one_H)$ be groups.
A \emph{homomorphism} $\phi\colon G\to H$ is a mapping satisfying 
$\phi(g\cdot g')=\phi(g)\ast \phi(g')$.
If $\phi\colon G\to H$ is a homomorphism, the set $\{h\in H\colon h=\phi(g) \text{ for some } g\in G\}$ is the \emph{image} of $\phi$, and the set $\{g\in G\colon \phi(g)=\one_H\}$ is the \emph{kernel} of $\phi$.

A non-empty subset $H$ of a group $G$ is a \emph{subgroup} of $G$, if 
$H$ is a group under the operation of $G$.
If $H$ is a subset of a group $G$, then the smallest subgroup of $G$ containing $H$ is \emph{the subgroup of $G$ generated by $H$}.

For non-empty subsets $H$, $K$ of a group $G$, define
$HK$ to be $HK=\{hk\colon h\in H \text{ and } k\in K\}$.
If $K=\{k\}$ we write $Hk$ for $H\{k\}$.
Let $H$ be a subgroup of a group $G$, and let $g\in G$; then $Hg$ ($gH$) is a 
\emph{right coset} (\emph{left coset}) of $H$ in $G$, and $g$ is a \emph{representative} of $Hg$ and $gH$.
If $H$ is a subgroup of $G$, the number of right cosets of $H$ is the same as the number of left cosets of $H$. The \emph{index} of $H$ in $G$ is the number of right (or left) cosets of $H$ in $G$.
A subgroup $H$ of $G$ is \emph{normal} if $gHg^{-1}$ is a subset of $H$ for all $g\in G$.

\subsection{Transformations}

A {\em transformation} of a set $Q$ is a mapping of $Q$ into itself. 
We consider only transformations of finite non-empty sets and, 
without loss of generality, assume that $Q=Q_n = \{0,1,\dots, n-1\}$. If $t$ is a transformation of $Q_n$
and  $i \in Q_n$, then $t(i)$ is the image of $i$ under $t$.  
An arbitrary transformation is written in the form
\begin{equation*}\label{eq:transmatrix}
t=\left( \begin{array}{ccccc}
0 & 1 &   \dots &  n-2 & n-1 \\
i_0 & i_1 &   \dots &  i_{n-2} & i_{n-1}
\end{array} \right),
\end{equation*}
where $i_k = t(k)$, $0\le k\le n-1$, and $i_k\in Q_n$. 
The {\em composition} of two transformations $t_1$ and $t_2$ of $Q_n$ is a transformation $t_1 \circ t_2$ such that $(t_1 \circ t_2)(i)= t_1( t_2(i))$ for all $i \in Q_n$. We usually omit the composition operator and write $t_1t_2$. 
The set of all transformations of $Q_n$ is a monoid with the identity transformation as the unit and composition as the operation.
A \emph{permutation} of $Q_n$ is a mapping of $Q_n$ \emph{onto} itself. 
In this paper we are concerned only with permutations.
The \emph{identity} transformation is denoted by $\one$.

A permutation $t$ is a \emph{cycle of of length $k$} or a  \emph{$k$-cycle} , where $k \ge 2$, if there exist pairwise different elements $i_1$,~\dots,~$i_k$ such that
$t(i_1)=i_2$, $t(i_2)=i_3$, \dots, $t(i_{k-1})=i_k$, and $t(i_k)=i_1$, and $t$ does not affect any other elements.
A~cycle is denoted by $(i_1,i_2,\dots,i_k)$.
A \emph{transposition} is a 2-cycle.  
Every permutation is a product (composition) of transpositions, and the parity of the number of transpositions in the factorization is an invariant. A permutation is \emph{odd} (\emph{even}) if its factorization has an odd (even) number of factors.

The \emph{symmetric group} $S_n$ of \emph{degree} $n$ is the set of all permutations of $Q_n$, with composition as the group operation, and the identity transformation as $\one$.
The \emph{alternating group} $A_n$ is the set of all even permutations of $S_n$.

Given a subgroup $H$ of $S_n$, we say that $H$ \emph{acts transitively} on $Q_n$ if for each $i,j\in Q_n$ there is some $t\in H$ such that $t(i)=j$.  We say that $H$ \emph{acts doubly transitively} on $Q_n$ if whenever $i,j,k,\ell\in Q_n$ with $i\neq j$ and $k\neq \ell$ there is some $t\in H$ such that $t(i)=k$, $t(j)=\ell$.

\subsection{Semiautomata and Automata}
A \emph{deterministic finite semiautomaton (DFS)} is a quadruple $\cA=(Q, \Sig, \delta, q_0)$, where 
$Q$ is the set of \emph{states}, $\Sig$ is a finite non-empty \emph{alphabet}, $\delta\colon Q\times \Sig\to Q$ is the \emph{transition function}, and $q_0$ is the \emph{initial state}. We extend $\delta$ to $Q \times \Sig^*$ in the usual way.
A state $q$ is  \emph{reachable} from the initial state if  there is a word $w$ such that 
$q=\delta(q_0,w)$. 
A DFS is \emph{connected} if every state $q\in Q$ is reachable.

For a  DFS $\cA=(Q,\Sig, \delta,q_0)$ and a word $w \in \Sigma^*$, the transition function $\delta(\cdot, w)$  is a transformation of $Q$, the transformation \emph{induced by $w$}. 
The set of all transformations induced by non-empty words is the \emph{transition semigroup} $S_\cA$ of $\cA$.
For $w\in\Sig^+$, we denote by $w\colon t$ the transformation $t$ of $Q_n$ induced by $w$.

Given semiautomata $\cA=(Q,\Sig, \delta,q_0)$ and $\cA'=(Q',\Sig, \delta',q'_0)$,
we define their direct product to be the DFS
$\cA\times \cA' =(Q\times Q',\Sig, (\delta,\delta') ,(q_0,q'_0))$.

A~\emph{deterministic finite automaton (DFA)} is a quintuple $\cD=(Q, \Sig, \delta, q_0,F)$, where $(Q,\Sig, \delta,q_0)$ is a DFS and 
$F\subseteq Q$ is the set of \emph{final states}. 
The DFA $\cD$ \emph{accepts} a word $w \in \Sigma^*$ if ${\delta}(q_0,w)\in F$. 
The set of all words accepted by $\cD$ is the \emph{language} $L(\cD)$ of $\cD$. 
The \emph{language accepted from a state} $q$ of a DFA is the language $L_{q}(\cD)$ accepted by the DFA
$(Q, \Sigma, \delta, q,F)$.
Two states of a DFA are \emph{distinguishable} if there exists a word $w$
which is accepted from one of the states and rejected from the other. Otherwise,
the two states are \emph{equivalent}.
A DFA is \emph{minimal} if all of its states are reachable from the initial state and no two states are equivalent. 
Note that if $|Q|\ge 2$ and $\cD$ is minimal, then $\emp\subsetneq F\subsetneq Q$.

\subsection{An Earlier Result}
Let  $\Sig=\{a,b\}$ and
$\cD=(Q_m, \Sig, \delta, 0,\{m-1\})$, where  $a\colon (0,\dots,m-1)$ and $b\colon (0,1)$, and let $L$ be the language of $\cD$.
Similarly,  let
$\cD'=(Q_n, \Sig, \delta', 0,\{n-1\})$, where  $a\colon (0,\dots,n-1)$ and $b\colon (0,1)$, and let $L'$ be the language of $\cD'$.
Also,  let
$\cD''=(Q_n, \Sig, \delta'', 0,\{n-1\})$, where  $b\colon (0,\dots,n-1)$ and $a\colon (0,1)$, and let $L''$ be the language of $\cD''$.

Let $\circ$ denote union, intersection, difference, or symmetric difference.
The following results were proved in~\cite{Brz12,BrLiu12}:
\begin{proposition}
\label{prop:BrzLiu}
For $L$, $L'$ and $L''$ as above and $m,n\ge 3$,  (a) the  complexity of $L\circ L''$ is $mn$, and (b) if $m\neq n$, the  complexity of $L\circ L'$ is $mn$.
 \end{proposition}
 
Our main theorem is a generalization of this result.
\section{Connectedness}
\label{sec:connect}

From now on we are interested in semiautomata $\cA$ and $\cA'$ whose transition semigroups are symmetric groups generated by two-element bases.
We assume that permutations $s$ and $s'$ are induced by  $a$ in $\cA$ and $\cA'$, and permutations $t$ and $t'$, by  $b$, that is,
$a\colon s$, $b\colon t$ in $\cA$, and $a\colon s'$, $b\colon t'$ in $\cA'$.

\begin{example}
Let $\Sig=\{a,b\}$, $\cA=(Q_3,\Sig,\delta,0)$, and $\cA'=(Q_3,\Sig,\delta',0)$, where
$a\colon s=(0,1,2)$, $b\colon t=(0,1)$ in $\cA$, 
and $a\colon s'=(0,1,2)$, $b\colon t'=(1,2)$ 
in $\cA'$.
Then $(s,t)$ and $(s',t')$ are conjugate, since $rsr^{-1}=s'$ and
$rtr^{-1}=t'$ for $r=(0,1,2)$.
On the other hand, if  $s''=(0,1)$ and $t''=(0,1,2)$,  then $(s,t)$ and $(s'',t'')$ are not conjugate.

The transition semigroups of $\cA$, $\cA'$ and $\cA''$ all have 6 elements.
Those of $\cA$ and $\cA'$, when viewed as semigroups generated by $a$ and $b$, are identical, but those of $\cA$ and $\cA''$ are not: for example,
$a^3=\one$ in $\cS_\cA$, but $a^2=\one$ in $\cS_{\cA''}$.
\qedb
\end{example}

\begin{theorem}
\label{thm: reach}
Let  $\Sig=\{a,b\}$, let $\cA=(Q_m,\Sig, \delta,0)$ and $\cA'=(Q_n,\Sig, \delta',0)$ be 
semiautomata with transition semigroups  that are  symmetric groups  of degrees $m$ and $n$ respectively, and let the corresponding bases be $B$ and $B'$.  For $m,n\ge 1$, the direct product $\cA\times \cA'$ is connected if and only if either (1) $m\neq n$ or (2) $m=n$ and $B$ and $B'$ are not conjugate. 
\end{theorem}

\begin{proof}
Without loss of generality, assume that $m\le n$.
Let $H$ denote the transition semigroup of $\cA\times \cA'$; then $H$ is a subgroup of $S_m\times S_n$. Define homomorphisms $\pi_1\colon H \to S_m$ and $\pi_2\colon H \to S_n$ by $\pi_1((s,t)) = s$ and $\pi_2((s,t)) = t$. Observe that $\pi_1$ and $\pi_2$ are surjective, since the transition semigroups of $\cA$ and $\cA'$ are $S_m$ and $S_n$ respectively.  We let $H_0$ denote the subgroup of $H$ consisting of all elements that map the set $\{0\}\times Q_n$ to itself.  Then $H_0$ has index $m$ in $H$ and thus $\pi_2(H_0)$ has index at most $m$ in $\pi_2(H)=S_n$.
Thus the order of $\pi_2(H_0)$ is at least $n!/m\ge (n-1)!$.  

Since a subgroup of $S_n$ that does not act transitively on $Q_n$ is necessarily isomorphic to a subgroup of $S_i\times S_{n-i}$ for some $i\in \{1,\ldots ,n-1\}$~\cite[Section 2.5.1]{Wil09},  a subgroup of $S_n$ whose order is strictly greater than $(n-1)!$ acts transitively on $Q_n$.  Moreover, a subgroup of order $(n-1)!$ that does not act transitively on $Q_n$ is isomorphic to $S_1\times S_{n-1}$; that is, it is the stabilizer of a point.  Thus $\pi_2(H_0)$ fails to act transitively on $Q_n$ if and only if $m=n$ and $\pi_2(H_0)$ is the stabilizer of a point. 

Suppose that $m< n$ or $m=n$ and $\pi_2(H_0)$ is not the stabilizer of a point,  which is equivalent to assuming that $\pi_2(H_0)$ acts transitively on $Q_n$.  We claim that the direct product $\cA\times \cA'$ is connected.  To see this, notice that given $(i,j)$ and $(i',j')$ in $Q_m\times Q_n$, we can find $t$ (respectively $t'$) in $H$ that  sends $(i,j)$ to $(0,k)$ (respectively $(i',j')$ to $(0,k')$) for some $k$ (respectively $k'$) in $Q_n$, since $\pi_1(H)=S_m$ acts transitively on $Q_m$.  
Since we have assumed that $\pi_2(H_0)$ acts transitively on $Q_n$, we can find $t'' \in H$ such that $\pi_2(t'')\in \pi_2(H_0)$  sends $(0,k)$ to $(0,k')$. Hence 
$(t')^{-1}t''t$ sends $(i,j)$ to $(i',j')$, and so $\cA\times \cA'$ is connected.

Suppose next that $m=n$ and $\pi_2(H_0)$ is the stabilizer of a point.   By relabelling if necessary, we may assume that $\pi_2(H_0)$ stabilizes $0$.   Then $H$ cannot send $(0,0)$ to $(0,i)$ for $i\neq 0$ and so $\cA\times \cA'$ is not connected.  We claim that the bases $B$ and $B'$ are conjugate.  

To prove this claim, note that $H$ has the property that if $(s,t)\in H\subseteq S_n\times S_n$ and $s(0)=0$, then $t(0)=0$.  
We claim there is a permutation $u \in S_n$ with $u(0)=0$ such that if $(s,t)\in H$ sends $(0,0)$ to $(j,k)$, then $k=u(j)$. First suppose that $k_1,k_2\in Q_n$ have the property that there is some $j\in Q_n$ such that $(j,k_1)$ and $(j,k_2)$ are in the orbit of $(0,0)$ under the action of $H$.  Then we can pick $h$ in $H$ such that $\pi_1(h)(j)=0$.  Then $(0,\pi_2(h)(k_1))$ and $(0,\pi_2(h)(k_2))$ are both in the orbit of $(0,0)$, which means that $\pi_2(h)(k_1)=\pi_2(h)(k_2) =0$, giving $k_1=k_2$.  It follows that there is a map $u\colon Q_n\to Q_n$ with $u(0)=0$ such that,
if $(s,t)\in H$ sends $(0,0)$ to $(j,k)$, then $k=u(j)$.  Since $\pi_2(H)$ acts transitively on $Q_n$, the map $u$ must be surjective and hence is a permutation, as claimed.

Let $s_1,s_2\in S_n$ denote the elements in the transition semigroup corresponding to $a\in \Sigma$, and let $t_1,t_2\in S_n$  correspond to $b\in \Sigma$.   
Let $H'$ be the group generated by $(s_1, u^{-1} t_1 u), (s_2,u^{-1} t_2  u)$.
Then $H'$ is conjugate to $H$ (we conjugate $H$ by $ ({\bf 1},u)$ to obtain $H'$); furthermore, $H'$  has the property that if $(s,t)\in H'$ sends $(0,0)$ to $(i,j)$, then $i=j$.  Thus $H'$ acts transitively on the diagonal of $Q_n\times Q_n$; if $(s,t) \in H'$ then $s(i)=t(i)$ for all $i\in Q_n$, which gives that $s=t$. Hence, if $(s,t')\in H$, then  $u^{-1}t' u=s$ and so the bases $B$ and $B'$ are conjugate.  
Thus if   
$\cA\times \cA'$ is not connected, then $m=n$ and the bases $B$ and $B'$ are conjugate.

 Now we show the converse: If $m=n$ and the bases $B=(s,t)$ and $B'=(s',t')$ are conjugate, then $\cA\times \cA'$ is not connected.
If  $rsr^{-1}=s'$, and  $rtr^{-1}=t'$, 
let $\psi_r\colon \{s,t\}^+\to \{s',t'\}^+$ be the mapping that assigns to  $x\in \{s,t\}^+$ the element 
$rxr^{-1}\in \{s',t'\}^+$.
For any $x,y\in \{s,t\}^+$, if $xy=z$, then 
$\psi_r(x)\psi_r(y)=(rxr^{-1})(ryr^{-1})=r(xy)r^{-1}=\psi_r(z)$.
Hence  the transition semigroups of $\cA$ and $\cA'$ are isomorphic.

The direct product $\cA\times \cA'$ is defined by $(Q_n\times Q_n,\{a,b\},
(\delta,\delta'),(0,0))$, where  $(\delta,\delta')((i,j),a) = (s(i),rsr^{-1}(j))$ and
$(\delta,\delta')((i,j),b) = (t(i),rtr^{-1}(j))$ for any $i,j\in Q_n$.  If $\cA\times \cA'$ is connected, then for all $(i,j)\in
Q_n\times Q_n$ there must exist a word $w\in \Sigma^{+}$ such that
$(\delta,\delta')((0,0),w)=(i,j)$ or, equivalently, there exists a permutation
$p$ such that $p(0)=i$ and $rpr^{-1}(0)=j$. 
There are now two cases:
\be 
\item
If $r^{-1}(0) \not= 0$, we prove that state $(i,r(i))$ is unreachable for all $i \in  Q_n$.
If $(i,r(i))$ is reachable, then there exists a permutation $p$ such that
$p(0)=i$ and $rpr^{-1}(0)=r(i)$. But then $r^{-1}rpr^{-1}(0)=pr^{-1}(0)=i=p(0)$, and so
$p^{-1}pr^{-1}(0)=r^{-1}(0)=0$, which is a contradiction.
\item
If $r^{-1}(0)= 0$, we prove that state $(i,i)$ is unreachable for some $i \in  Q_n$.
Since $r$ cannot be the identity, there must exist an $i$ such 
that $r(i)\not=i$. Suppose $(i,i)$ is reachable for that $i$.
Then there exists a permutation
$p$ such that $p(0)=i$ and $rpr^{-1}(0)=i$. 
Thus $i=rpr^{-1}(0)=rp(0)=p(0)$ and 
$r(i)=i$, which is a contradiction.
\ee

In either case   $\cA\times \cA'$ is not connected. 
\qed
\end{proof}

\begin{remark}
\label{rem:sc}
If $\cA\times \cA'$ is connected, then it is strongly connected, since the transition semigroup of $\cA\times \cA'$ is a group.
\end{remark}

\section{Uniformly Minimal Semiautomata}
\label{sec:uniform}
Semiautomata that result in minimal DFAs under any non-trivial assignment of final states were studied by Restivo and Vaglica~\cite{ReVa12}.
We modify their definitions slightly to suit our purposes.
A strongly connected DFS $\cA=(Q,\Sig,\delta,q_0)$ with $|Q|\ge 2$ is \emph{uniformly minimal}
if the DFA $\cD=(Q,\Sig,\delta,q_0,F)$ is minimal for each set $F$ of final states, where $\emp\subsetneq F\subsetneq Q$.

Given a DFS $\cA=(Q,\Sig,\delta,q_0)$, we define the \emph{pair graph}
of $\cA$ to be the directed graph $G_\cA=(V_\cA,E_\cA)$, where the set $V_\cA$ of vertices  is the set of all two-element subsets $\{p,q\}$ of $Q$, and the set $E_\cA$ of edges consists of unordered pairs $(\{p,q\},\{p',q'\})$ such that 
$\{\delta(p,a),\delta(q,a)\}=\{p',q'\}$.
The following result was proved in~\cite{ReVa12}:

\begin{proposition}[Restivo and Vaglica]
\label{prop:Restivo}
Let $\cA=(Q,\Sig,\delta,q_0)$ be a strongly connected DFS with at least two states. If the pair graph
$(V_\cD,E_\cD)$ is strongly connected, then $\cA$ is uniformly minimal.
\end{proposition}

We prove a similar result for semiautomata   with transition semigroups  that are the symmetric groups. 

\begin{proposition}
\label{prop:DFS}
Suppose that $\cA=(Q_n,\Sig,\delta,q_0)$ is a DFS and the transition semigroup $S_\cA$ of $\cA$ is the symmetric group $S_n$. Then $\cA$ is strongly connected and uniformly minimal.
\end{proposition}
\begin{proof}
If  $S_\cA=S_n$, then $S_\cA$ contains all permutations of $Q_n$, in particular, the cycle $(0,\dots,n-1)$; hence $\cA$ is strongly connected. For any $(i,j), (k,\ell)\in Q_n\times Q_n$ with $i\neq j$, $k\neq \ell$, and $\{i,j\}\neq \{k,\ell\}$,  any permutation that maps $i$ to $k$ and $j$ to $\ell$ connects $\{i,j\}$ to $\{k,\ell\}$ in the pair graph of $\cA$. Hence the pair graph is strongly connected, and $\cA$ is uniformly minimal by Proposition~\ref{prop:Restivo}.
\qed
\end{proof}

Let the truth values of propositions be 1 (true) and 0 (false). Let $\circ\colon \{0,1\}\times \{0,1\} \to \{0,1\}$ be a binary boolean function.
Extend $\circ$ to a function
$\circ\colon 2^{\Sig^*}\times 2^{\Sig^*}\to 2^{\Sig^*}$:
If $w\in \Sig^*$ and $L,L'\subseteq \Sig^*$, 
then $$w\in (L\circ L') \Leftrightarrow (w\in L) \circ (w\in L').$$
Also, extend $\circ$ to a function
$\circ\colon 2^{Q_m}\times 2^{Q_n}\to 2^{Q_m\times Q_n}$:
If $q\in Q_m$, $q'\in Q_n$, $F\subseteq Q_m$, and $F'\subseteq Q_n$, 
then $$(q,q')\in (F\circ F') \Leftrightarrow (q\in F) \circ (q'\in F').$$

Suppose that $\cA=(Q,\Sig, \delta,0)$ and $\cA'=(Q',\Sig, \delta',0)$ with $|Q|=m$ and $|Q'|=n$ are uniformly minimal DFSs, and 
$\circ$ is any proper boolean function.
The pair  $(\cA, \cA')$  is \emph{uniformly minimal for $\circ$}
if the direct product $\cP=(Q\times Q', \Sig, (\delta,\delta'), (0,0), F\circ F')$
is minimal for all valid assignments of $F$ and $F'$ of sets of final states to $\cA$ and $\cA'$, that is, sets $F$ and $F'$ such that $\emp\subsetneq F\subsetneq Q$ and $\emp\subsetneq F'\subsetneq Q'$.

If $n=1$,  then $\cA\times\cA'$ is isomorphic to $\cA$ and  no
boolean function $\circ$ is proper.
Hence this case, and also the case with $m=1$, is of no interest.
Henceforth we assume that  $m,n\ge 2$.

We now consider  pair graphs of DFSs with symmetric groups as their transition semigroups.  

\begin{example}
\label{ex:22}
Suppose now that $m=n=2$, and $\cA$ and $\cA'$ both have $S_2$ as their transition semigroup.
There are two permutations in $S_2$: $(0,1)$ and $\one$, and 
there are  three bases:  $B_1=(a\colon (0,1), b\colon (0,1))$,
$B_2=(a\colon (0,1), b\colon \one)$, and
$B_3=(a\colon \one, b\colon (0,1))$.
Note that no two of these bases are conjugate.

For each basis, there are two possible final states, 0 or 1, and  hence  two DFAs; thus there are six different DFAs.
There are then twelve direct products $\cD^i_j\times \cD^k_\ell$ with non-conjugate bases, where
$\cD^i_j$ ($\cD^k_\ell$) uses basis $B_i$ ($B_k$) and has $j$ ($\ell$) as final state, for $i,k=1,2,3$ and $j,\ell=1,2$.

For each pair of DFAs accepting languages $L$ and $L'$ respectively, we tested the complexity of five boolean functions: 
$L\cup L'$,  $L\cap L'$, $L\oplus L'$ , $L\setminus L'$ and $L'\setminus L$. Note that the complexity of each remaining proper boolean function is the same as that of one of these five functions.
For all twelve direct products of DFAs with non-conjugate bases, all proper boolean functions reach the maximal complexity 4, except for the functions $L\oplus L'$ and $\ol{L\oplus L'}$, which fail in all twelve cases.
Thus any two DFAs $\cD=(Q_2,\Sig,\delta_i,0,F)$ and $\cD'=(Q_2,\Sig,\delta_k,0,F')$, where
$Q_2=\{0,1\}$, $\Sig=\{a,b\}$, $\delta_i$ ($\delta_k$) is defined by basis $B_i$ ($B_k$),
$F=\{j\}$ and $F'=\{\ell\}$, are uniformly minimal for all proper boolean functions, except $\oplus$ and its complement. So our main result applies only in some cases if $m=n=2$.
\qedb
\end{example}

\begin{proposition}
\label{prop:pairgraph}
Let $\cA=(Q_m,\Sig,\delta,0)$ and $\cA'=(Q_n,\Sig,\delta',0)$, with $m,n \ge 2$ and $\max(m,n)\ge 3$, be 
DFSs with transition semigroups  that are symmetric groups, and let $\cP$ be their direct product.
Then the following hold:
\be
\item
The pair graph of $\cP$ consists of strongly connected components---which we will call simply \emph{components}---of one of the following three types:
\bi
\item
$T_1\subseteq C_1=\{\{(i,j), (k,\ell)\}\colon i\neq k, j\neq \ell\}$, 
\item
$T_2\subseteq C_2=\{\{(i,j), (i,\ell)\}\colon j\neq \ell\}$, 
\item
$T_3\subseteq C_3=\{\{(i,j), (k,j)\}\colon i\neq k\}$. 
\ei
\item
Every state $(i,j)$ of the direct product $\cP$ appears in at least one pair in each component.\item
Each component has at least $mn/2\ge 3$ pairs.
\ee
\end{proposition}
\begin{proof}
The first claim follows since the transition semigroup of $\cP$ is a group.
The second claim holds because the direct product is strongly connected, by Remark~\ref{rem:sc}.
For the third claim, note that there are $mn$ states in $\cP$, but they can appear in pairs; hence the bound $mn/2$. Since we are assuming that $mn\ge 6$, the last claim follows.
\qed
\end{proof}

Now consider DFAs $\cD=(Q_m,\Sig,\delta,0,F)$ and $\cD'=(Q_n,\Sig,\delta',0,F')$, where
$\emp\subsetneq F \subsetneq Q_m$ and $\emp\subsetneq F' \subsetneq Q_n$.
A~state $\{(i,j), (k,\ell)\}$ of the pair graph of the direct product $\cP$ of $\cD$ and $\cD'$ is \emph{distinguishing} if and only if $(i,j)$ is final and $(k,\ell)$ is not, or \emph{vice versa}.

\begin{example}
\label{ex:23}
Suppose $m=2$ and $n=3$, $\cD$ and $\cD'$ are as above, $\delta$ is defined by the the basis  $(a\colon \one,b \colon (0,1))$ of $S_2$, and $\delta'$ by the basis $(a\colon(0,1,2), b\colon (0,1))$ of $S_3$.
The direct product of $\cD$ and $\cD'$ is connected as guaranteed by Theorem~\ref{thm: reach}, and has six states. The
 components in the pair graph are:
\bi
\item
$C_{1,1}= \{ \mathbf{\{(0,0), (1,1)\}}, \{(0,1), (1,2)\}, \{(0,2), (1,0)\} \}$, 
\item
$C_{1,2}=\{ \{(0,0), (1,2)\}, \mathbf{\{(0,1), (1,0)\}}, \{(0,2), (1,1)\} \}$,
\item
$C_2=\{ \{(0,0), (0,1)\}, \mathbf{\{(0,1), (0,2)\}}, \mathbf{\{(0,0), (0,2)\}}, \{(1,0), (1,1)\}, \\ \mbox{\hspace{1.1cm}}  \mathbf{\{(1,0), (1,2)\}}, \mathbf{\{(1,1), (1,2)\} \}}$, 
\item
$C_3=\{ \mathbf{\{(0,0), (1,0)\}}, \mathbf{\{(0,1), (1,1)\}}, \mathbf{\{(0,2), (1,2)\}} \}$.
\ei
One verifies that if $F=\{0\}$, $F'=\{0,1\}$ and the boolean function is symmetric difference,  the distinguishing pairs are  in boldface. 
We return to this case in Section~\ref{sec:small}.\qedb
\end{example}

\begin{example}
\label{ex:34}
If $m=3$ and $n=4$, $\delta$ is defined by the basis $(a\colon (0,1),b \colon (0,1,2))$ of $S_3$,  $\delta'$ by  the basis $(a\colon(0,1),b\colon (1,3,2))$ of $S_4$. One verifies that these bases are not conjugate.
The direct product $\cP$  is connected and has twelve states.

If $F=\{2\}$, $F'=\{0,1\}$ and intersection is the boolean function,
the component of the pair graph containing $\{(0,0),(0,3)\}$ is:\\
\mbox{\hspace{1cm} } $T= \{ \{(0,0), (0,3)\}, \{(0,1), (0,2)\}, \{(1,0), (1,2)\},
$\\
$\mbox{\hspace{1.85cm} }  \{(1,1), (1,3)\}, (2,0), (2,1)\}, \{(2,2), (2,3)\} \},$\\
and there are no distinguishing pairs. Hence states $(0,0)$ and $(0,3)$ are equivalent in $\cP$, as are also any two states appearing in the same pair of $T$.
Indeed, the minimal version of $\cP$ has exactly six states.
For symmetric difference, there are only four states, but there are twelve states for union.
Here our theorem applies only in some cases if $m=3$ and $n=4$.
\qedb
\end{example}

\begin{example}
\label{ex:44}
Suppose $m=n=4$,  $\delta$ is defined by the basis $(a\colon (0,1,2), b\colon (2,3))$, and  $\delta'$ by  the basis
$(a\colon(1,3,2), b\colon (0,2,1,3))$. If $F=\{0,1\}$ and
$F'=\{0,1\}$,
then the complexity of $L\oplus L'$ is 4, but all the other complexities are 12.
The same holds if $F=\{0,3\}$ and $F'=\{1,2\}$.
Again, our theorem applies only in some cases if $m=n=4$.
\qedb
\end{example}


\begin{lemma}
\label{lem:dist}
Let $\cD=(Q,\Sig,\delta,0,F)$ and $\cD'=(Q',\Sig,\delta',0,F')$, with $|Q|,|Q'| \ge 2$, be 
DFAs with transition semigroups  that are groups, and let $\cP=(Q\times Q', \Sig, (\delta,\delta'), (0,0), F\circ F')$ be their direct product.
Then $\cP$ is minimal if and only if every component of the pair graph $G_\cP$ of $\cP$ has a distinguishing pair.
\end{lemma}
\begin{proof}  Let $H$ be the transition group of the direct product $\cP=\cD\times\cD'$.  
Suppose $s\in H$ corresponds to the transformation of $Q\times Q'$ induced by some word $w\in \Sig^+$; 
then for $(i,j)\in Q\times Q'$, define $s\cdot (i,j)$ to be $(\delta(i,w),\delta'(j,w))$.

Suppose first that every component of $G_\cP$ has a distinguishing pair, but $\cP$ is not minimal.   Then there must be two distinct states $(i,j),(k,\ell)\in Q\times Q'$ such that, for $s\in H$, $s\cdot (i,j)$ is a final state if and only if $s\cdot (k,\ell)$ is also final.
By assumption, there is a distinguishing pair $\{(i',j'),(k',\ell')\}$ in the component of $G_\cP$ that contains $\{(i,j),(k,\ell)\}$.  By interchanging $(i',j')$ and $(k',\ell')$ if necessary, we may assume that there is some $s\in H$ such that 
$s\cdot (i,j) = (i',j')$ and $s\cdot (k,\ell)=(k',\ell')$.  But this is a contradiction.

Conversely, suppose there is a component $C$ without a distinguishing pair. Then, if 
$(i,j)$ and $(k,\ell)$ appear in the same pair, they must be equivalent since they can only reach states that are both final or both non-final. 
\qed
\end{proof}

\section{Symmetric Groups and Complexity of Boolean Operations}
\label{sec:main}

We begin with a well-known but apparently unpublished result.
\begin{lemma}
Let $n$ be a positive integer, let $G$ be either $S_n$ or $A_n$, and let $H$ be a subgroup of $G$ of index $m\le n$.  Then the following hold:
\begin{enumerate}
\item[(i)] if $n\neq 4$ and $m<n$, then $H$ is either $A_n$ or $S_n$;
\item[(ii)] if $m=n$ and $n\neq 6$, then there is some $i\in Q_n$ such that $H$ is the set of permutations in $G$ that fix $i$.
\item[(iii)] if $m=n=6$, then there is an automorphism $\phi$ of $S_6$ such that $\phi(H)$ is the set of elements that fix $0$.\end{enumerate}
\label{lem: bertrand}
\end{lemma}
\begin{proof} For $n\le 3$, both (i) and (ii) are clear.  Thus assume that $n\ge 4$.  
Let $X=\{H=x_0H,\dots ,x_{m-1}H\}$ be the set of left cosets of $H$ in $G$.   Note that $G$ acts on $X$ via left multiplication; more explicitly,  for $g\in G$, there is a permutation $s\in S_m$ such that $gx_i H = x_{s(i)}H$ for all $i\in \{0,\dots ,m-1\}$.  The map $g\mapsto s$ gives a non-trivial homomorphism $\phi$ from $G$ into $S_m$.  Furthermore, the kernel of $\phi$ is necessarily contained in $H$, since the kernel of $\phi$ is $\{g\in G\colon gx_iH=x_iH~{\rm for~all~}i=0,\dots ,m-1\}$ and this is contained in $\{g\in G\colon gH=H\}=H$.    

If $m<n$, then $|G|>|S_m|$ and hence $\phi$ must have a non-trivial kernel, which is a normal subgroup of $G$.  For $n\ge 5$, the only normal subgroups of $G$ are either $A_n$ or $S_n$. Since the kernel of $\phi$ is a normal subgroup contained in $H$,  $H$ must be either $A_n$ or $S_n$, if $n\ge 5$.   This establishes (i).

On the other hand, if $m=n$ and $n\not\in \{4,6\}$,  we have a non-trivial homomorphism $\phi\colon G\to S_n$.  If the kernel is non-trivial, then  again $H$ must be $A_n$ or $S_n$, which contradicts the fact that $H$ has index $n$.  If the kernel is trivial, then $\phi$ gives an embedding of $G$ into $S_n$.  If $G=S_n$ then $\phi$ is an automorphism.  If $G=A_n$ then the image of $\phi$ is an index-two subgroup of $S_n$ and hence necessarily $A_n$.  Thus $\phi$ gives an automorphism of $G$ in either case.  For $n\neq 6$, all automorphisms of $S_n$ or $A_n$ are given by conjugation by an element of $S_n$ (see \cite[Chapter 3.2]{Suz82}).  Since $h \in H$ stabilizes the coset $H=x_0H$, the definition of the map $\phi$ now gives $\phi(h)(0)=0$, and so $\phi(H)$ stabilizes 0.  Since $\phi$ is given by conjugation by an element of $S_n$, we see that $H$ consists of all elements of $G$ that stabilize some $i\in Q_n$.  Thus we have proved that (ii) holds except when $n=4$.  This argument also gives (iii) immediately.  

If $m=n=4$, as before we have a non-trivial homomorphism $\phi:G\to S_4$ and the kernel must be one of $S_4$, $A_4$, $K_4$ (the Klein 4-group), or the trivial group.  Since the kernel of $\phi$ is contained in $H$ and $H$ has order $3$ or $6$,  the kernel is in fact trivial and $\phi$ is an embedding.   The argument used above now proves (ii) in this case.  
\qed  
\end{proof}

The following lemma, like Theorem~\ref{thm: reach}, deals with reachability. The conditions in the lemma, however, are useful for determining reachability in the pair graph of $\cA\times\cA'$, rather than in  $\cA\times\cA'$ itself.
\begin{lemma}
Let  $\Sig=\{a,b\}$, let $\cA=(Q_m,\Sig, \delta,0)$ and $\cA'=(Q_n,\Sig, \delta',0)$ be 
semiautomata with transition semigroups  that are  symmetric groups  of degrees $m$ and $n$ respectively with $m\le n$, $n\neq 4$ and $(m,n)\neq (6,6)$.  
Let $H$ be the transition semigroup of $\cA\times \cA'$, and let $\pi_1$ and $\pi_2$ be the natural projections from $H$ onto $S_m$ and $S_n$ respectively. If
$H_0=\{h\in H\colon\pi_1(h)(0)=0\},$ 
then 
\be
\item $\pi_2(H_0)$ is either $S_n$ or $A_n$, or is the stabilizer of a point in $Q_n$. 
\item
$\pi_2(H_0)$ is the stabilizer of a point if and only if $m=n$, and in this case the direct product $\cA\times \cA'$ is not connected.
\ee
\label{lem: H0}
\end{lemma}
\begin{proof}
For Part 1, 
since $\pi_1(H)=S_m$, for each $i\in \{0,\dots ,m-1\}$ there is some $h_i\in H$ such that $\pi_1(h_i)$ takes $0$ to $i$.  For a given $h\in H$, $\pi_1(h)$ takes $0$ to $j$ for some $j\in \{0,1,\dots ,m-1\}$, and thus $h_j^{-1}h\in H_0$ and so $h\in h_j H_0$.  However, since $\pi_1(h)$ takes $0$ to $j$, we have $h_i^{-1}h \not\in H_0$ and thus $h \not\in h_iH_0$ for $i \ne j$. Thus the cosets $h_0H,\ldots ,h_{m-1}H$ are distinct, and  $H_0$ has index $m$ in $H$.  Since 
$$\pi_2(H)\subseteq \bigcup_{i=0}^{m-1} \pi_2(h_i)\pi_2(H_0),$$ 
$\pi_2(H_0)$ has index at most $m$ in $\pi_2(H)=S_n$.   If $n\neq 4$ and $m<n$ then $\pi_2(H_0)$ is either $A_n$ or $S_n$ by Lemma \ref{lem: bertrand}. 
If $m=n$ and $n\neq 6$, then $\pi_2(H_0)$ has index $n$ in $S_n$ and hence must be the stabilizer of a some $i\in Q_n$ by Lemma \ref{lem: bertrand}.  

For Part 2, suppose that $m=n$ and $\pi_2(H_0)$ is the stabilizer of a point in $Q_n$.
By relabelling if necessary, we may assume that $\pi_2(H_0)$ stabilizes $0$.  Hence, if $h\in H$ sends $(0,0)$ to $(0,j)$ then $j=0$.  In particular, there is no $h\in H$ that sends $(0,0)$ to $(0,1)$ or that sends $(0,1)$ to $(0,0)$ and so $\cA\times \cA'$ is necessarily not connected.
\qed
\end{proof}

\begin{lemma} 
\label{lem:main}
Let $\cA=(Q_m,\Sig, \delta,0)$ and $\cA'=(Q_n,\Sig, \delta',0)$ be 
semiautomata with transition semigroups  that are the symmetric groups  of degrees $m$ and $n$, respectively with $m\le n$, $m\ge 2$, $n\ge 5$, and $(m,n)\neq (6,6)$.  If $\cA\times \cA'$ is connected then the pair graph of $\cA\times \cA'$ has exactly three connected components:
$C_1=\{\{(i,j), (k,\ell)\}\colon i\neq k, j\neq \ell\}$, $C_2=\{\{(i,j), (i,\ell)\}\colon j\neq \ell\}$, and 
$C_3=\{\{(i,j), (k,j)\}\colon i\neq k\}$. 

\end{lemma}
\begin{proof}
We let $H$ denote the transition semigroup of $\cA\times \cA'$.  In addition to this, we let 
$C_1=\{\{(i,j), (k,\ell)\}\colon i\neq k, j\neq \ell\}$, $C_2=\{\{(i,j), (i,\ell)\}\colon j\neq \ell\}$, and 
$C_3=\{\{(i,j), (k,j)\}\colon i\neq k\}$.  We show that each of $C_1,C_2,C_3$ is strongly connected.  Note that each of $C_1$, $C_2$, $C_3$ is necessarily a union of connected components. 

We show that $C_1$ is strongly connected.  Suppose we have pairs $\{(i,j), (k,\ell)\}$ and $\{(i',j'), (k',\ell')\}$ with $i,k$ distinct, $i',k'$ distinct, $j,\ell$ distinct, and $j',\ell'$ distinct.  Since $S_m$ acts doubly transitively on $Q_m$ when $m\ge 2$,  there is some $s\in H$ that sends $(i,j)$ to $(i',j'')$ and $(k,\ell)$ to $(k',\ell'')$ for some $j'',\ell''\in Q_n$.

Thus we may assume without loss of generality that $i'=i$ and $k'=k$.  Let $H_0$ be the subgroup of $S_m\times S_n$ consisting of all $x \in H$ such that $\pi_1(x)$ fixes $i$.  By Lemma \ref{lem: H0},  since we assume that $\cA\times\cA'$ is connected,
$\pi_2(H_0)$ is not a stabilizer of a point in $Q_n$.
Hence $\pi_2(H_0)$ is either $S_n$ or $A_n$.  
Let $H_1$ denote the subgroup of $S_m\times S_n$ consisting of all $x\in H$ such that $\pi_1(x)$ fixes $i$ and $k$.  
By the argument used in Lemma \ref{lem: H0} to show that $\{h \in H\colon \pi_1(h)(0) = 0\}$ has index $m$ in $H$, we see that $\pi_2(H_1)$ has index at most $m-1$ in $\pi_2(H_0)$.  
Thus $\pi_2(H_1)$ is a subgroup of $A_n$ or $S_n$ of index at most $n-1$, and hence must again be $A_n$ or $S_n$ by Lemma \ref{lem: bertrand}.  
Since $A_n$ and $S_n$ both act doubly transitively on $Q_n$, there is some $h\in H$ that sends $(i,j)$ to $(i,j')$ and $(k,\ell)$ to $(k,\ell')$ whenever $\ell$ and $\ell'$ are distinct.  This proves that $C_1$ is indeed a strongly connected component.

Next, consider pairs $\{(i,j), (i,k)\}$ with $j,k$ distinct.  For given $\{(i',j'),(i',k')\}$ with $j',k'$ distinct, there is some element $s\in H$ such that $\pi_1(s)(i)=i'$ and thus $s$ sends
$(i,j)$ to $(i',j'')$ and $(i,k)$ to $(i',k'')$ for some $j'',k''\in Q_n$ with $j''\neq k''$.  Now note that $\pi_2(\{x\in H\colon\pi_1(x)(i')=i'\})$ is either $S_n$ or $A_n$ by Lemma \ref{lem: H0}, and thus acts doubly transitively on $Q_n$.  It follows that there is some $s'\in H$ such that
$s'$ sends $(i',j'')$ to $(i',j')$ and $(i',k'')$ to $(i',k')$.  Then $s's$ sends $\{(i,j), (i,k)\}$ to $\{(i',j'),(i',k')\}$ and thus $C_2$ is strongly connected.

Finally, consider pairs $\{(i,j), (k,j)\}$ and $\{(i',j'),(k',j')\}$ with $i,k$ distinct and $i',k'$ distinct.  From the argument used in proving $C_1$ is strongly connected, we see that we can find $s\in H$ that sends
$\{(i,j), (k,j)\}$ to $\{(i',j''), (k',j'')\}$ for some $j''$.  As in the proof that $C_1$ is strongly connected,  we see that the image of the set of $h\in H$ for which $\pi_1(h)$ stabilizes both $i'$ and $k'$ under $\pi_2$ acts transitively on $Q_n$; hence we can find $s'\in H$ that  sends $\{(i',j''), (k',j'')\}$ to $\{(i',j'),(k',j')\}$.  Thus $C_3$ is strongly connected.
\qed
\end{proof} 

We are now in a position to prove our main result in all except a few cases which are dealt with in Section~\ref{sec:small}.

\begin{corollary}
Let $m$ and $n$ be positive integers with $n\ge m \ge 2$, $n\ge 5$, and $(m,n)\neq (6,6)$, and let $\cA=(Q_m,\Sig, \delta,0)$ and $\cA'=(Q_n,\Sig, \delta',0)$ be 
semiautomata with transition semigroups  that are the symmetric groups  of degrees $m$ and $n$. Suppose that the direct product $\cA\times \cA'$ is connected and assume further that sets of final states are added to $\cA$ and $\cA'$ and that $\circ$ is a proper binary boolean function that defines the set of final states of the direct product $\cP$. Then $\cP$ is minimal for any such $\circ$.
\label{cor:main}

\end{corollary}
\begin{proof} 
By Lemma~\ref{lem:main}, the pair graph of $\cA\times \cA'$ has three strongly connected components:
$C_1=\{\{(i,j), (k,\ell)\}\colon i\neq k, j\neq \ell\}$, $C_2=\{\{(i,j), (i,\ell)\}\colon j\neq \ell\}$, and 
$C_3=\{\{(i,j), (k,j)\}\colon i\neq k\}$.  

For $(i,j)\in Q_m\times Q_n$, define $f((i,j))$ to be $1$ if $(i,j)$ is a final state, and  $0$, otherwise.  We first claim that $C_1$ has a distinguishing pair, that is, there are pairs $(i,j)$ and $(k,\ell)$ in $Q_m\times Q_n$ with $i\neq k$ and $j\neq \ell$ such that $f((i,j))\neq f((k,\ell))$.  

Suppose no distinguishing pair exists in $C_1$.
Assume without loss of generality that $f((0,0))=0$.
 then $f((i,j))=0$ whenever $i\neq 0$ and $j\neq 0$.  
Given $k\in Q_n$, we pick $\ell\in Q_n\setminus \{0,k\}$; this is always possible since $n\ge 3$.  
Since $\{(0,k),(1,\ell)\}$ is in $C_1$ and we have assumed that $C_1$ has no distinguishing pairs, we must have $f((0,k))=f((1,\ell))$.
But $f(1,\ell)$ must be 0, for otherwise we would have the distinguishing pair $\{(0,0),(1,\ell)\}$.
Hence $f((0,k))=f((1,\ell))=0$.   
Thus we have
$f((i,j))=0$ for every $i\in Q_m$ and every $j\in Q_n\setminus \{0\}$.  
Similarly, we must have $f((i,0))=f((0,1))=0$ for $i\in Q_m\setminus \{0\}$,  and hence $f$ is the zero function, a contradiction.  

The fact that $C_2$ and $C_3$ both have distinguishing pairs follows from the fact that $\circ$ is a proper boolean function.
By Lemma~\ref{lem:dist}, we conclude that  $\cA\times \cA'$ is uniformly minimal.
\qed
\end{proof}

\section{Results for Small Values of $m$ and $n$}
\label{sec:small}
We have proved our main result in the case that $m\le n$ and $n\ge 5$ if $(m,n)\neq (6,6)$.  By symmetry we may always assume that $m\le n$.  The case $(m,n)=(2,2)$ was handled in Example~\ref{ex:22}, that of $(m,n)=(3,4)$, in Example~\ref{ex:34}, and that
of $(m,n)=(4,4)$, in Example~\ref{ex:44}. Therefore the only cases that we need to consider are those with $(m,n)\in \{(2,3),(2,4),(3,3),(6,6)\}$.  

In this section we prove the following result:
\begin{theorem}
\label{thm: small}
Let $\cA=(Q_m,\Sig, \delta,0)$ and $\cA'=(Q_n,\Sig, \delta',0)$ be 
semiautomata with transition semigroups that are  $S_m$ and $S_n$ respectively. Suppose that the direct product $\cA\times \cA'$ is connected, sets of final states are added to $\cA$ and $\cA'$, and  $\circ$ is a proper binary boolean function that defines the set of final states of the direct product $\cP$. If $(m,n)\in \{(2,3),(2,4),(3,3),(6,6)\}$, then $\cP$ is minimal for any such $\circ$.
\end{theorem}
The theorem is proved in four parts, since each case requires  a different argument.
The following remark, however, is common to all parts.

\begin{remark}
\label{rem:equiv}
If  there is a proper boolean function $\circ$ for which $\cP$ is not minimal, then there must be two distinct states $(i,j), (i',j')\in Q_m\times Q_n$  such that $s\cdot (i,j)$ is final if and only if $s\cdot (i',j')$ is final for any $s$ in the transition semigroup of~$\cP$.  

Define an equivalence relation on $Q_m \times Q_n$ by declaring that $(i,j)\sim (i',j')$ precisely when  $s\cdot (i,j)$ is final if and only if $s\cdot (i',j')$ is final for all $s$.  
This equivalence relation partitions $Q_m\times Q_n$ into disjoint parts.  Moreover, each equivalence class must have the same size, since $\cP$ is connected; in particular, each equivalence class has size equal to a divisor of $mn$.  If each part in the partition has size~$1$, then $\cP$ is minimal and there is nothing to prove.   
If there is exactly one part in the partition, then either all states of $\cP$ are final or all non-final; in either case $\circ$ is not proper, a contradiction. 
\end{remark}
\subsection{$(m,n)=(6,6)$}

\begin{lemma}
Let $\phi:S_6\to S_6$ be an outer automorphism.  
\begin{enumerate}
\item[(1)] Let $T$ be a subgroup of $S_6$ of order $120$. 
If $T$ has a fixed point, then  $\phi(T)$ does not, and 
if $\phi(T)$ has a fixed point, then $T$ does not.
\item[(2)] 
If $T$ is as in (1), then either $T$ or $\phi(T)$ has a fixed point. 
\item[(3)]
If $N$ is a subgroup of index $2$ in the stabilizer subgroup of some $i\in Q_6$, then $\phi(N)$ acts doubly transitively on $Q_6$.
\end{enumerate}
\label{lem:s6}
\end{lemma}

\begin{proof}
We first show (1).  If $T$ stabilizes some point, then $T$ contains a transposition since it has order 120.  
Since $\phi$ is outer, it sends any transposition to a product of three disjoint transpositions~\cite{MR0096724}. 
Since the product of three disjoint transpositions has no fixed points,  $\phi(T)$ cannot have a fixed point.  
Similarly,  if $\phi(T)$ has a fixed point then $\phi^2(T)$ cannot have a fixed point.  
Since $\phi^2$ is inner~~\cite[p. 133]{Rot65}, $\phi^2(T)$ is conjugate to $T$ and hence $T$ cannot have a fixed point.  

For (2), we must show that at least one of $T$ and $\phi(T)$ fixes some point.  
Suppose that $T$ does not have a fixed point.  
Then $S_6$ acts on the left cosets of $T$ by left multiplication, and this gives a map $\psi\colon S_6\to S_6$ (we think of the copy of $S_6$ on the right-hand side as acting on cosets of $T$).  
Note that the kernel of $\psi$ is contained in $T$, and since $A_6$ and $S_6$ are the only non-trivial normal subgroups of $S_6$, 
 the kernel of $\psi$ is trivial and so $\psi$ is an automorphism.  
 Note that $\psi(T)$ stabilizes the coset $T$ by definition of our map and hence $\psi(T)$ has a fixed point in $S_6$.  
 Since $T$ does not have a fixed point, 
 $\psi$ cannot be an inner automorphism.  
 Since the inner automorphism group of $S_6$ has index $2$ in the full automorphism group~\cite[p. 133]{Rot65}, 
  $\phi$ can be obtained by composing $\psi$ with an inner automorphism and so $\phi(T)$ has a fixed point.  Thus we have shown that if $T$ has no fixed point,  then  $\phi(T)$ does. 
 It follows that exactly one of $T$ and $\phi(T)$ has a fixed point.

We now prove (3).  We first show that $\phi(N)$ acts transitively on $Q_6$.  If it did not, then $\phi(N)$ would be contained in a conjugate of a subgroup of $S_6$ of the form $S_i\times S_{6-i}$ for some $i\in \{1,2,3\}$.  Since $|\phi(N)|=60\ge i! (6-i)!$ for $i=2,3$, we see that $\phi(N)$ would necessarily fix some $j\in Q_6$.  By replacing $\phi$ by $\phi$ composed with some appropriate inner automorphism, we may assume that our outer automorphism $\phi$ has the property that both $\phi(N)$ and $N$ fix $i\in Q_6$.  This means that $\phi(N)$ and $N$ are both contained in the stabilizer subgroup, $H$, of $i\in Q_6$, which is a group of order $120$.  We claim that $\phi(H)=H$.  If not, then $\phi(N)$ is normal of index $2$ in both $H$ and $\phi(H)$, and so by the second isomorphism theorem~\cite[p. 26]{Rot65}, $H\phi(H)$ generates a group of order $240$ in $S_6$.  But this is impossible by Lemma \ref{lem: bertrand}, since this group would have index $3$ in $S_6$.  It follows that $H=\phi(H)$, which contradicts (1), since both $H$ and $\phi(H)$ fix $i$.  Thus $\phi(N)$ acts transitively on $Q_6$.  

To show that $\phi(N)$ acts doubly transitively, it suffices to prove that the set of elements of $Q_6$ that stabilize $i\in Q_6$ acts transitively on $Q_6\setminus \{i\}$.   The orbit of $i$ under the action of $\phi(N)$ has size $6$; hence the stabilizer is an index-$6$ subgroup of $\phi(N)$ and thus has size $10$.   
Since it has size $10$, it must contain a 5-cycle $s$ on the elements $Q_6\setminus \{i\}$, and for given $j,k\in Q_6\setminus \{i\}$, we have that $s^m\cdot j=k$ for some $m\ge 0$.  
It follows that $\phi(N)$ acts doubly transitively on $Q_6$.  
\qed
\end{proof}
\begin{proposition}
Theorem~\ref{thm: small} holds for $(m,n)=(6,6)$.
\end{proposition}
\begin{proof}  Let $H$ denote the transition semigroup of $\cP$.  Then $H$ is a subgroup of $S_6\times S_6$.  We let $\pi_1$ and $\pi_2$ denote the two natural surjections from $H$ onto $S_6$.  For $i\in Q_6$, let $H_i$ denote the subgroup of $S_6$ obtained by applying $\pi_2$ to the collection of $x\in H$ such that $\pi_1(x)$ fixes $i$.  Then $H_i$ has index at most $6$ in $S_6$ and hence must be one of $A_6$, $S_6$, or a group of order $120$.  If each $H_i$ is either $A_6$ or $S_6$, then we may follow the argument of Lemma \ref{lem:main} to show that the pair graph of $\cA\times \cA'$ has exactly three connected components; namely,
$C_1=\{\{(i,j), (k,\ell)\}\colon i\neq k, j\neq \ell\}$, $C_2=\{\{(i,j), (i,\ell)\}\colon j\neq \ell\}$, and 
$C_3=\{\{(i,j), (k,j)\}\colon i\neq k\}$.  Then the argument from Corollary \ref{cor:main} shows that $\cP$ is minimal whenever $\circ$ is a proper binary boolean function.

After relabelling if necessary, it is sufficient to consider the case that $H_0$ is a group of order $120$.   Let $N=\{s\in S_6\colon (\one, s)\in H\}$.  Then $N$ is a normal subgroup of $S_6$ and hence must be one of  $A_6$, $S_6$, or the trivial subgroup.  
Since $N\subseteq H_0$, $N$ must be trivial.  
If we define $\phi$ to be $\phi=\pi_2\circ \pi_1^{-1}\colon S_6\to S_6$,
then $\phi$ is an automorphism, and so $H=\{(s,\phi(s))\colon s\in S_6\}$.   Since $\cP$ is connected,  $\phi$ cannot be an inner automorphism by Theorem~\ref{thm: reach}.  We claim that in this case, $\cP$ is necessarily minimal for any proper boolean function.

Suppose  there is a proper boolean function $\circ$ for which $\cP$ is not minimal.  
Define the equivalence $\sim$ as in Remark~\ref{rem:equiv};
then each equivalence class has size equal to a divisor of $36$, and we can ignore the cases where that size is 1 or 36.    

Let $E = \{(t,\phi(t))\colon t {\rm ~stabilizes}~ 5\}$.  Then $E$ has size $120$ and $\pi_2(E)$ acts transitively on $Q_6$, since $\phi$ is outer.  Let $F=\{x \in E \colon \pi_2(x) {\rm ~stabilizes}~ 5\}$; then $F$ has size $20$. 
Since $F$ has size $20$, it contains an element of the form $(t,\phi(t))$ where both $t$ and $\phi(t)$ are $5$-cycles that permute $\{0,1,2,3,4\}$.  It follows that $\pi_1(F)$ and $\pi_2(F)$ both act transitively on $\{0,1,2,3,4\}$. 

Now let $X\subseteq Q_6\times Q_6$ denote the equivalence class of $(5,5)$; then $X$ has size at least $6$.  Since $|X|$ divides $36$, we see that $|X|\in \{6,9,12,18, 36\}$, but, as noted before, 36 can be ignored.  

We now do a case-by-case analysis.  

\item[Case 1: $|X| = 18$.]

In this case  the orbit of $X$ under $H$ has size $2$.  Let $N$ denote the set of elements of $H$ that stabilize $X$.  Then $N$ has index $2$ in $H$, and hence must be equal to $\{(s,\phi(s))\colon s \in A_6\}$.

We now claim that if  $(i,j)$ and $(i,j')$ are in $X$ for some $i$ and distinct $j$, $j'$, then $(i,k)$ is in $X$ for all $k$.  To see this, observe that the set of $s$ in $N$ for which $\pi_1(s)\cdot i=i$ has index $6$ in $N$, and so it is a subgroup of order $60$.  
Thus it is of the form $(N_1,\phi(N_1))$ where $N_1$ is the copy of $A_5$ inside the set of elements of $\pi_1(H)$ that stabilize $i$, which is isomorphic to $S_5$.   
Notice that $\phi(N_1)$ acts doubly transitively on $Q_6$ by Lemma \ref{lem:s6}, since $\phi$ is an outer automorphism of $S_6$, and so we get the result.  A similar result holds for $(Q_6 \times Q_6) \setminus X$, which means that, for a fixed $i$, the set of $k$ for which $(i,k)$ is in $X$ is either $Q_5$ or empty, and so our boolean function is a function of the first variable, a contradiction.

\item[Case 2. $|X|\in \{9,12\}$.]

In this case, the orbit of $X$ under $H$ has size either $3$ or $4$ and thus the stabilizer of $X$ has index $3$ or $4$ in $H$.  But $H$ is isomorphic to $S_6$ and hence has no subgroups of index $3$ or $4$ by Lemma \ref{lem: bertrand} (i).

\item[Case 3. $|X|=6$.]
By the remarks above, we have 
$$X=\{(5,5),(0,\tau(0)),\ldots ,(4,\tau(4))\}$$ for some permutation $\tau$ of $\{ 0,1,2,3,4\}.$

The orbit of $X$ under $H$ has size $36/|X| = 6$, and so the stabilizer, $N$, of $X$ in $H$ has size $120$.  This means  that $N = (T, \phi(T))$ where $T$ is a subgroup of $S_6$ of order $120$.  Either $T$ or $\phi(T)$ must have a fixed point by Lemma \ref{lem:s6}.  Without loss of generality, $T$ has a fixed point and $\phi(T)$ does not.  We extend $\tau$ to a permutation of $Q_6$ by declaring that $\tau(5)=5$.  If $T$ fixes $i$, we have $N\cdot (i,\tau(i)) = (i,\tau(i))$, and so $\phi(T)$ fixes $\tau(i)$, a contradiction.  The result follows.
\qed
\end{proof}

\subsection{$(m,n)=(3,3)$}
\begin{lemma}   
Suppose that $H$ is a $2$-generated subgroup of $S_3\times S_3$ with the property that the two natural projections into $S_3$ are surjective.  Then either $H=S_3\times S_3$ or there is some permutation $t\in S_3$ such that $H=\{(s,t^{-1}st)\colon s\in S_3\}$.  
\label{lem:s3}
\end{lemma}
\begin{proof}
Let $N=\{s\colon (\one,s)\in H\}$.  Then $N$ is a normal subgroup of $S_3$ and hence must be one of $A_3$, $S_3$, or the trivial subgroup.  If $N$ is trivial then the first projection is an isomorphism and hence $H$ is isomorphic to $S_3$.  Thus $\pi_2\circ \pi_1^{-1}:S_3\to S_3$ is an automorphism of $S_3$.  Since all automorphisms of $S_3$ are inner, there exists some $t\in S_3$ such that $H=\{(s,tst^{-1})\colon s\in S_3\}$.  

If $N=A_3$, then we also know that $N'=\{s\colon (s,\one)\in H\}$ is $A_3$.  Thus if $(s,t)\in H$, $s$ and $t$ are either both even permutations or both odd permutations.  Every generating set for $S_3$ must contain a transposition and thus one generator of $H$ must be of the form $(s,t)$ with $s$ and $t$ transpositions.  By conjugating by a permutation in the second coordinate, we may assume that our first generator is $(s,s)$ for some transposition $s$.  Let $(t,u)$ be the second generator for $H$.  Then either $t$ and $u$ are both 3-cycles or they are both transpositions not equal to $s$.  In both cases, either $u=t$ or $u=sts$. By conjugating by either $(\one,\one)$ or  $(\one,s)$, we see that it is no loss of generality to assume that $H$ is generated by $(s,s)$ and $(t,t)$ for two elements of $S_3$.  But this contradicts the fact that $N=A_3$.

If $N=S_3$, then $H$ has size $36$ and hence must be $S_3\times S_3$. 
\qed
\end{proof}

\begin{proposition}
Theorem~\ref{thm: small} holds for $(m,n)=(3,3)$.
\end{proposition}
\begin{proof}
Let $H$ denote the transition semigroup of $\cP$.  Since $\cA\times \cA'$ is connected,  $H=S_3\times S_3$ by Lemma \ref{lem:s3}.

Suppose  there is a proper boolean function $\circ$ for which $\cP$ is not minimal.  
Each equivalence class of Remark~\ref{rem:equiv} has size equal to a divisor of $9$, and we can ignore the cases where that size is 1 or 9; hence the size must be 3. 

Let $X\subseteq Q_3\times Q_3$ be the part in the partition of $Q_3\times Q_3$ that contains $(0,0)$ and let $(i,j)\neq (0,0)$ be another element of $X$.  
Since $|X|=3$, there exists $(k,\ell)\in \{1,2\}\times \{1,2\}$ that is not in $X$.  
If $i\neq 0$ and $j\neq 0$, then $H=S_3\times S_3$ acts doubly transitively on $Q_3\times Q_3$. 
Hence there exists $s\in H$ such that $s\cdot (0,0)=(0,0)$ and $s\cdot (i,j)=(k,\ell)$, which is a contradiction since either $s\cdot X = X$ or $(s\cdot X)\cap X$ is empty.   
We conclude that if $(i,j)\in X$, then either $i=0$ or $j=0$.  
Next suppose that $X$ contains an element of the form $(0,j)$ with $j\neq 0$  and an element of the form $(i,0)$ with $i\neq 0$.  
Then
$X=\{(0,0),(0,j),(i,0)\}$.  
If we let $(s,t)\in H$ be the pair in which $s$ is the identity and $t$ is a 3-cycle that sends $0$ to $j$, then  $((s,t)\cdot X)\cap X$ has size 1, a contradiction, since it is either all of $X$ or empty.  We conclude that $X$ is either $\{0\}\times Q_3$ or $Q_3\times \{0\}$.  But then $\circ$ is a constant function and hence not proper.   The result follows.
\qed
\end{proof}
\subsection{$(m,n)=(2,3)$}
\begin{proposition}
Theorem~\ref{thm: small} holds for $(m,n)=(2,3)$.
\end{proposition}
\begin{proof}
Let $H$ denote the transition semigroup of $\cP$.  Then the natural projections from $H$ to $S_2$ and $S_3$ are both surjective.  In particular, $H$ has size either $6$ or $12$, and it can be verified that it contains all elements of the form $(s,t)$ in which $s$ and $t$ are either both even or both odd.  

Suppose  there is a proper boolean function $\circ$ for which $\cP$ is not minimal.  
Each equivalence class of Remark~\ref{rem:equiv} has size equal to a divisor of $6$, and we can ignore the cases where that size is 1 or 6; hence the size must be 2 or 3.    

Let $X\subseteq Q_2\times Q_3$ be a part in our partition.  
If there exist $i\in Q_2$ and distinct $j,k\in Q_3$ such that $(i,j),(i,k)\in X$, then by relabelling we may assume that $i=0$ and $j=0$, $k=1$.  
Since $u=(\one, (0,1,2))\in H$, and $u\cdot (0,0)=(0,1)$, we see that $u\cdot X=X$ and so $X\supseteq \{0\}\times Q_3$.  
Since $|X|\le 3$,  the partition consists of the two parts $\{0\}\times Q_3$ and $\{1\}\times Q_3$, contradicting the fact that $\circ$ is proper.  
It follows that $|X|=2$ and each part of our partition consists of an element of the form $(0,i)$ and an element of the form $(1,j)$ for some $i,j\in Q_3$.  We cannot have $i=j$, since then  $\circ$ would be a constant function.  By relabelling if necessary, we may assume that $X_0=\{(0,0),(1,1)\}$ is one part of our partition.  Letting $u=(\one,(0,1,2))$ act on $X$, we see that
$X_1=\{(0,1),(1,2)\}$ and $X_2=\{(0,2),(1,0)\}$ are the remaining parts that make up our partition of $Q_2\times Q_3$.  It is no loss of generality to assume that exactly two elements of $Q_3$ are final states.  Let $i,j$ be these two final states of $Q_3$. Since $(0,i)$ and $(0,j)$ are either both final or both non-final, either all states in $X_i\cup X_j$ are final or none of them are.  
It is no loss of generality to assume that all states of $X_i\cup X_j$ are final.  
But  $(1,i+1)$ and $(1,j+1)$ are both final, where $i+1$ and $j+1$ are taken modulo $3$.  Then also $(1,i)$ and $(1,j)$ are either both final or both non-final.  Since, modulo $3$, $\{i,j\}\cap \{i+1,j+1\}$ is non-empty and $\{i,j,i+1,j+1\}=\{0,1,2\}$,  all states of the form $\{1\}\times Q_3$ are final, and thus all states are final, contradicting that $\circ$ is proper.  The result follows.
\qed
\end{proof}
\section{$(m,n)=(2,4)$}
\begin{proposition}
Theorem~\ref{thm: small} holds for $(m,n)=(2,4)$.
\end{proposition}
\begin{proof}
Let $H$ denote the transition semigroup of $\cP$.  Then the natural projections from $H$ to $S_2$ and $S_4$ are both surjective and so $H$ has size either $24$ or $48$, and  contains all elements of the form $(s,t)$ in which $s$ and $t$ are either both even or both odd.

Suppose there is a proper boolean function $\circ$ for which $\cP$ is not minimal.  
Each equivalence class of Remark~\ref{rem:equiv} has size equal to a divisor of $8$, and we can ignore the cases where that size is 1 or 8; hence the size must be 2 or 4. 

Let $X\subseteq Q_2\times Q_4$ be a part in our partition.  
If there exist $i\in Q_2$ and distinct $j,k\in Q_4$ such that $(i,j),(i,k)\in X$, then by relabelling we may assume that $i=0$ and $j=0$, $k=1$.  
Since $u=(\one, (0,1,2))\in H$, and $u\cdot (0,0)=(0,1)$, we see that $u\cdot X=X$ and so $X\supseteq \{0\}\times \{0,1,2\}$.  
Similarly, $v=(\one, (0,1,3))\in H$, and $v\cdot (0,0)=(0,1)$; hence $v\cdot X=X$ and so $X\supseteq \{0\}\times \{0,1,3\}$.  
Thus $X\supseteq \{0\}\times Q_4$.  
Since $|X|\le 4$, our partition must consist of the two parts $\{0\}\times Q_4$ and $\{1\}\times Q_4$, which contradicts the fact that $\circ$ is proper.  

Thus  $|X|=2$ and each part of our partition consists of an element of the form $(0,i)$ and an element of the form $(1,j)$ for some $i,j\in Q_4$.  We  cannot have $i=j$, since then  $\circ$ would be a constant function.  Thus, by relabelling if necessary, we may assume that $X=\{(0,0),(1,1)\}$ is one part of our partition.  But $u=((0,1),(0,1,2,3))\in H$, and since $u\cdot (0,0)=(1,1)$, we see that $u\cdot X=X$.  However $u\cdot (1,1)=(0,2)\not\in X$, a contradiction.  The result follows.
\qed
\end{proof}
\section{Conclusions}
\label{sec:conc}
We have shown that if the inputs of two DFAs induce transformations that constitute non-conjugate bases of symmetric groups, then the quotient complexity of all non-trivial boolean operations on the languages accepted by the DFAs is maximal, except for a few special cases when the sizes of the DFAs are small. We believe that other similar results are possible and deserve further study.
\smallskip 

\noin
{\bf Acknowledgment}
We thank Gareth Davies for his careful reading of our manu\-script and for his constructive comments.

\providecommand{\noopsort}[1]{}

\end{document}